\documentclass[11pt,twoside,a4paper]{article}

\usepackage{epsfig}
\usepackage{amsfonts}
\usepackage{amssymb}
\usepackage{amstext}
\usepackage{amsmath}
\usepackage{xspace}
\usepackage{shadow,times}
\usepackage{algorithm,algorithmic}
\usepackage{fullpage}

\newcommand{\ignore}[1]{}

\allowdisplaybreaks

\newtheorem{theorem}{Theorem}
\newtheorem{lemma}[theorem]{Lemma}
\newtheorem{claim}[theorem]{Claim}

\newenvironment{proof}{

\noindent{\bf Proof:}} {\hfill$\blacksquare$

}

\newcommand{\initOneLiners}{%
    \setlength{\itemsep}{0pt}
    \setlength{\parsep }{0pt}
    \setlength{\topsep }{0pt}
}

\newenvironment{OneLiners}[1][\ensuremath{\bullet}]
    {\begin{list}
        {#1}
        {\initOneLiners}}
    {\end{list}}

\newcommand{\sse}{\subseteq}
\def\lrp{$k$-LocVRP\xspace}
\def\opt{{\sf Opt}\xspace}

\newcommand{\MST}{\ensuremath{\textrm{MST}}}
\def\kmf{$k$ median forest\xspace}

\title{Locating Depots for Capacitated Vehicle Routing}
\author{
Inge Li G{\o}rtz\thanks{Technical University of Denmark.} \and Viswanath Nagarajan\thanks{IBM T.J. Watson Research
Center. } }
\date{}

\begin{document}

\maketitle
\begin{abstract}
We study a location-routing problem in the context of capacitated vehicle routing. The input to \lrp is a set of demand
locations in a metric space and a fleet of $k$ vehicles each of capacity $Q$. The objective is to {\em locate} $k$
depots, one for each vehicle, and {\em compute routes} for the vehicles so that all demands are satisfied and the total
cost is minimized. Our main result is a constant-factor approximation algorithm for \lrp. To achieve this result, we
reduce \lrp to the following generalization of $k$ median, which might be of independent interest. Given a metric
$(V,d)$, bound $k$ and parameter $\rho\in\mathbb{R}_+$, the goal in the {\em $k$ median forest} problem is to find
$S\sse V$ with $|S|=k$ minimizing:
$$\sum_{u\in V} d(u,S) \quad + \quad \rho\cdot d\big(\,\mbox{MST}(V/S)\,\big),$$
where $d(u,S)=\min_{w\in S} d(u,w)$ and $\mbox{MST}(V/S)$ is a minimum spanning tree in the graph obtained by
contracting $S$ to a single vertex. We give a $(3+\epsilon)$-approximation algorithm for $k$ median forest, which leads
to a $(12+\epsilon)$-approximation algorithm for \lrp, for any constant $\epsilon>0$. The algorithm for  $k$ median
forest is $t$-swap local search, and we prove that it has locality gap $3+\frac2t$; this generalizes the corresponding
result for $k$ median~\cite{AGKMMP04}.

Finally we consider  the $k$ median forest problem when there is a different cost function $c$ for the MST part, i.e.
the objective is $\sum_{u\in V} d(u,S) \,+ \,c (\,\mbox{MST}(V/S)\,)$. We show that the locality gap for this problem
is unbounded even under multi-swaps, which contrasts with the $c=d$ case. Nevertheless,  we obtain a constant-factor
approximation algorithm, using an LP based approach along the lines of~\cite{KKNSS11}.

\end{abstract}

\section{Introduction}\label{sec:intro}
In typical facility location problems, one wishes to locate centers and connect clients directly to centers at minimum
cost. On the other hand, the goal in vehicle routing problems (VRPs) is to compute routes for vehicles originating from
a given set of depots. Location routing problems represent an integrated approach, where we wish to  make combined
decisions on facility location and vehicle routing. This is a widely researched area in operations research, see eg.
surveys~\cite{BWW87,L88,L89,BJS95,MJS98,NS07}. Most of these papers deal with exact methods or heuristics, without any
performance guarantees. In this paper we present an approximation algorithm for a location routing problem in context
of capacitated vehicle routing.

Capacitated vehicle routing (CVRP) is an extensively studied vehicle routing problem~\cite{TV02} which involves
distributing identical items to a set of demand locations. Formally we are given a metric space $(V,d)$ on vertices $V$
with distance function $d:V\times V\rightarrow \mathbb{R}_+$ that is symmetric and satisfies triangle inequality. Each
vertex $u\in V$ demands $q_u$ units of the item. We have available a fleet of $k$ vehicles, each having capacity $Q$
and located at specified depots. The goal is to distribute items using the $k$ vehicles at minimum total cost. There
are two versions of CVRP depending on whether or not the demand at a vertex may be satisfied over multiple visits. We
focus on the  {\em unsplit delivery} version in the paper, while noting that this also implies the result under
split-deliveries.

We consider the question ``where should one locate the $k$ depots so that the resulting vehicle routing solution has
minimum cost?'' This is called {\em $k$-location capacitated vehicle routing} (\lrp). The \lrp problem bears obvious
similarity to the well-known {\em $k$ median} problem, where the goal is to choose $k$ centers to minimize the sum of
distances of each vertex to its closest center. The difference is that our problem also takes the routing aspect into
account.  Not surprisingly, our algorithm for \lrp builds on approximation algorithms for the $k$ median problem.

In obtaining an algorithm for \lrp we introduce the {\em  $k$ median forest} problem, which might be of some
independent interest. The objective here is a combination of $k$-median and minimum spanning tree. Given metric
$(V,d)$, bound $k$ and parameter $\rho\in\mathbb{R}_+$, the goal is to find $S\sse V$ with $|S|=k$ minimizing
$\sum_{u\in V} d(u,S) \, + \, \rho\cdot d\big(\,\mbox{MST}(V/S)\,\big)$.  Here $d(u,S)=\min_{w\in S} d(u,w)$ is the
minimum distance between $u$ and an $S$-vertex; $\mbox{MST}(V/S)$ is a minimum spanning tree in the graph obtained by
contracting $S$ to a single vertex. Note that when $\rho=0$ we have the $k$-median objective, and $\rho$ being very
large reduces to MST.
\subsection{Our Results}
The main result is the following.
\begin{theorem}\label{th:lcvrp}
There is a $(12+\epsilon)$-approximation algorithm for \lrp, for any constant $\epsilon>0$.
\end{theorem}
Our algorithm first reduces \lrp to $k$ median forest, at the loss of a constant approximation factor of four. This
step is fairly straightforward and makes use of
known lower-bounds~\cite{HK85} for the CVRP problem. We present this reduction in Section~\ref{sec:redn}.

Then we prove the following result in Section~\ref{sec:kmf} which implies Theorem~\ref{th:lcvrp}.
\begin{theorem}\label{th:kmed-forest}
There is a $(3+\epsilon)$-approximation algorithm for $k$ median forest, for any constant $\epsilon>0$.
\end{theorem}
This is the technically most interesting part of the paper. The algorithm is straightforward: perform local search
using multi-swaps. It is well known that (single swap) local search is optimal for the minimum spanning tree problem.
Moreover, Arya et al.~\cite{AGKMMP04} showed that $t$-swap local search achieves exactly a $(3+\frac2t)$-approximation
ratio for the $k$-median objective (this proof was later simplified by Gupta and Tangwongsan~\cite{GT08}). Thus one can
hope that local search performs well for $k$ median forest, which is a combination of both MST and $k$-median
objectives. However, the local moves used in proving the quality of local optima are different for the MST and
$k$-median objectives. Our proof shows we can {\em simultaneously} bound both MST and $k$-median objectives using a
common set of local moves. In fact we prove that the locality gap for \kmf under $t$-swaps is also $(3+\frac2t)$.
Somewhat surprisingly, it suffices to consider exactly the same set of swaps from~\cite{GT08} to establish
Theorem~\ref{th:kmed-forest}, although~\cite{GT08} does not take into account any MST contribution. The interesting
part of the proof is in bounding the change in MST cost due to these swaps--- this makes use of non-trivial exchange
properties of spanning trees and properties of the potential swaps from~\cite{GT08}. We remark that the $k$-median,
$k$-tree (i.e. choose $k$ centers $S$ to minimize $d(MST(V/S))$), and \kmf objectives are incomparable in general:
Appendix~\ref{app:example} gives an instance where near-optimal solutions to these three objectives are mutually far
apart.

Finally we consider the {\em non-uniform  $k$ median forest} problem in Section~\ref{sec:non-unif-kmf}. This is an
extension of $k$ median forest where there is a different cost function $c$ for the MST part in the objective. Given
vertices $V$ with two metrics $d$ and $c$, and bound $k$, the goal is to find $S\sse V$ with $|S|=k$ minimizing
$\sum_{u\in V} d(u,S) \, + \, c\big(\,\mbox{MST}(V/S)\,\big)$. Here $\mbox{MST}(V/S)$ is a minimum spanning tree in the
graph obtained by contracting $S$ to a single vertex, {\em under metric $c$}. In contrast to the uniform case $c=d$, we
show that the locality gap here is unbounded even for multi-swaps. In light of this, Theorem~\ref{th:kmed-forest}
appears a bit surprising. Still, we show that a different LP-based approach yields:
\begin{theorem}\label{th:gen-kmed-forest}
There is a 16-approximation algorithm for non-uniform $k$ median forest.
\end{theorem}
This algorithm follows closely that for the matroid median problem~\cite{KKNSS11}. We consider the natural LP
relaxation and round it in two phases. The first phase sparsifies the solution (using ideas from~\cite{CGTS99}) and
allows us to reformulate a new LP-relaxation using fewer variables; this is identical to~\cite{KKNSS11}. The second
phase solves the new LP-relaxation, which we show  to be integral.

\subsection{Related Work}
The basic capacitated vehicle routing problem involves a single fixed depot. There are two versions of CVRP: {\em split
delivery} where the demand of a vertex may be satisfied over multiple visits; and {\em unsplit delivery} where the
demand at a vertex must be satisfied in a single visit (in this case we also assume $\max_{u\in V} q_u\le Q$). Observe
that the optimal value under split-delivery is at most that under unsplit-delivery. The best known approximation
guarantee for split-delivery is $\alpha+1$~\cite{HK85,AG90} and for unsplit-delivery is $\alpha+2$~\cite{AG87}, where
$\alpha$ denotes the best approximation ratio for the Traveling Salesman Problem. We make use of the following known
lower bounds for CVRP with single depot $r$: the minimum TSP tour on all demand locations, and $\frac{2}Q\sum_{u\in V}
d(r,u)\cdot q_u$. Similar constant factor approximation algorithms~\cite{LS90} are also known for the CVRP with
multiple depots which was defined in the introduction.

The $k$ median problem is a widely studied location problem and has many constant factor approximation algorithms.
Starting with the LP-rounding algorithm of~\cite{CGTS99}, the primal-dual approach was used in~\cite{JV01}, and also
local search~\cite{AGKMMP04}. A simpler analysis of the local search algorithm was given in~\cite{GT08}; we make use of
this in our proof for the $k$ median forest problem. Several variants of $k$ median have also been studied. One that is
relevant to us is the matroid median problem~\cite{KKNSS11}, where the set of open centers are constrained to be
independent in some matroid; our approximation algorithm for the non-uniform $k$ median forest problem is based on this
approach.

Recently~\cite{HKM10} studied (among other problems) a facility-location variant of CVRP: there are opening costs for
depots and the goal is to open a set of depots and find vehicle routes so as to minimize the sum of opening and routing
costs. The \lrp problem in this paper can be thought of as the $k$-median variant of~\cite{HKM10}. In~\cite{HKM10} the
authors give a 4.38-approximation algorithm for facility-location CVRP. Following an approach similar to~\cite{HKM10}
one can obtain a bicriteria approximation algorithm for \lrp, where more than $k$ depots are opened. However more work
is needed to obtain a true approximation, and this is where we need an algorithm for the $k$ median forest problem.


\section{Reducing \lrp to $k$ median forest} \label{sec:redn}
\def\flow{{\sf Flow}}
\def\tree{{\sf Tree}}
\def\med{{\sf Med}}

Here we show that the \lrp problem can be reduced to $k$ median forest at the loss of a constant approximation factor.
This makes use of known lower bounds for CVRP~\cite{HK85,LS90,HKM10}.

For any $S\sse V$, let $\flow(S) := \frac2Q \, \sum_{u\in V} q_u\cdot d(u,S)$ and $\tree(S) = d(MST(V/S))$ be the
length of the minimum spanning tree in the metric obtained by contracting $S$. The following theorem is implicit in
previous work~\cite{HK85,LS90,HKM10}; this uses a natural MST splitting algorithm.
\begin{theorem}[\cite{HKM10}] \label{th:cvrp}
Given any instance of CVRP on metric $(V,d)$ with demands $\{q_u\}_{u\in V}$, vehicle capacity $Q$ and depots $S\sse
V$,
\begin{OneLiners}
 \item The optimal value (of the split-delivery CVRP) is at least $\max\{\flow(S),\,\tree(S)\}$.
 \item There is a polynomial time algorithm that computes an unsplit-delivery solution of length at most $2\cdot \flow(S)+2\cdot \tree(S)$.
\end{OneLiners}
\end{theorem}

Based on this it is clear that the optimal value of the CVRP instance given depot positions $S$ is roughly given by
$\flow(S)+\tree(S)$, which is similar to the \kmf objective. The following lemma formalizes this reduction. We will
assume an algorithm for the $k$ median forest problem with vertex-weights $\{q_u:u\in V\}$, where the objective becomes
$\sum_{u\in V} q_u\cdot d(u,S) \, + \, \rho\cdot d\big(\,\mbox{MST}(V/S)\,\big)$.

\begin{lemma}\label{lem:lrp2kmed}
If there is a $\beta$-approximation algorithm for $k$ median forest then there is a $4\beta$-approximation algorithm
for \lrp.
\end{lemma}
\begin{proof} Let \opt denote the optimal value of the \lrp instance. Using the lower bound in Theorem~\ref{th:cvrp},
$$\opt \ge \min_{S:|S|=k} \, \max\left\{ \flow(S), \, \tree(S)\right\} \ge \min_{S:|S|=k} \, \left[\epsilon\cdot \flow(S) + (1-\epsilon)\cdot \tree(S)\right],$$
where $\epsilon\in [0,1]$ is any value; this will be fixed later. Consider the instance of $k$ median forest on metric
$(V,d)$, vertex weights $\{q_u\}_{u\in V}$ and parameter $\rho=\frac{1-\epsilon}\epsilon \cdot \frac{Q}2$. For any
$S\sse V$ the objective is:
$$\sum_{u\in V} q_u\cdot d(u,S) + \rho\cdot d(MST(V/S)) = \frac{Q}2 \cdot \flow(S) + \rho\cdot \tree(S) =
\frac{Q}{2\epsilon} \cdot \left[\epsilon\cdot \flow(S) + (1-\epsilon)\cdot \tree(S)\right] $$ Thus the optimal value of
the \kmf instance is at most $\frac{Q}{2\epsilon} \cdot \opt$. Let $S_{alg}$ denote the solution found by the
$\beta$-approximation algorithm for \kmf. It follows that $|S_{alg}|=k$ and:
\begin{equation} \label{eq:lrp2kmed}
\epsilon\cdot \flow(S_{alg}) + (1-\epsilon)\cdot \tree(S_{alg}) \le \beta\cdot \opt\end{equation} For the \lrp
instance, we locate the depots at $S_{alg}$. Using Theorem~\ref{th:cvrp}, the cost of the resulting vehicle routing
solution is at most $2\cdot \flow(S_{alg}) + 2\cdot \tree(S_{alg}) = 4\cdot \left[\epsilon\cdot \flow(S_{alg}) +
(1-\epsilon)\cdot \tree(S_{alg}) \right]$ where we set $\epsilon= 1/2$. From Inequality~\eqref{eq:lrp2kmed} it follows
that our algorithm is a $4\beta$-approximation algorithm for \lrp.
\end{proof}

\medskip

We remark that this reduction already gives us a constant factor {\em bicriteria} approximation algorithm for \lrp as
follows. Let $S_{med}$ denote an approximate solution to $k$-median on metric $(V,d)$ with vertex-weights $\{q_u:u\in
V\}$, which can be obtained by directly using a $k$-median algorithm~\cite{AGKMMP04}. Let $S_{mst}$ denote the optimal
solution to $\min_{S: |S|\le k} \, d(MST(V/S))$, which can be obtained using the greedy MST algorithm. We output
$S_{bi}=S_{med}\bigcup S_{mst}$ as a solution to \lrp, along with the vehicle routes obtained from
Theorem~\ref{th:cvrp} applied to $S_{bi}$. Note that $|S_{bi}|\le 2k$, so we open at most $2k$ depots. Moreover, if
$S^*$ denotes the location of depots in the optimal solution to \lrp then:
\begin{OneLiners}
\item $\flow(S_{med})\le (3+\delta)\cdot \flow(S^*)$ since we used a $(3+\delta)$-approximation algorithm for
$k$-median~\cite{AGKMMP04}.
\item $\tree(S_{mst})\le \tree(S^*)$ since  $S_{mst}$ is an optimal solution to the MST part of the objective.
\end{OneLiners}
Clearly $\flow(S_{bi}) \le \flow(S_{med})$ and $\tree(S_{bi}) \le \tree(S_{mst})$, so:
$$\frac12\cdot \flow(S_{bi}) + \frac12\cdot \tree(S_{bi}) \le \frac{3+\delta}2 \cdot \left[
\flow(S^*) + \tree(S^*)\right] \le (3+\delta)\cdot \opt$$ Using Theorem~\ref{th:cvrp} the cost of the CVRP solution
with depots $S_{bi}$ is at most $4(3+\delta)\cdot \opt$. So this gives a $(12+\delta,\, 2)$ bicriteria approximation
algorithm for \lrp, where $\delta>0$ is any fixed constant. We note that this approach combined with algorithms for
facility-location and Steiner tree immediately gives a constant factor approximation for the facility location CVRP
considered in~\cite{HKM10}. The algorithm in that paper~\cite{HKM10} has to do  some more work in order to get a
sharper constant. For \lrp this approach clearly does not give any true approximation ratio, and for this purpose we
give an algorithm for \kmf.


\section{Multi-swap local search for \kmf} \label{sec:kmf}
\def\swap{\ensuremath{\mathcal{S}}\xspace}
\def\p{\ensuremath{\mathcal{P}}\xspace}
\def\f{\ensuremath{\mathcal{F}}\xspace}
\def\C{\ensuremath{\mathcal{C}}\xspace}

The input to {\em $k$ median forest} consists of a metric $(V,d)$, vertex-weights $\{q_u\}_{u\in V}$ and bound $k$. The
goal is to find $S\sse V$ with $|S|=k$ minimizing:
$$\Phi(S)=  \sum_{u\in V} q_u\cdot d(u,S) \quad + \quad d\big(\,\MST(V/S)\,\big),$$
where $d(u,S)=\min_{w\in S} d(u,w)$ and $\MST(V/S)$ is a minimum spanning tree in the graph obtained by contracting $S$
to a single vertex. Note that this is slightly more general than the definition in Section~\ref{sec:intro} (which is
the special case when $q_u=1/\rho$ for all $u\in V$).

We analyze the natural $t$-swap local search for this problem, for any constant $t$. Starting at an arbitrary solution
$L$ consisting of $k$ centers do the following until no improvement is possible: if there exists $D\sse L$ and $A\sse
V\setminus L$ with $|D|=|A|\le t$ and $\Phi(L\setminus D \bigcup A) < \Phi(L)$ then $L\gets L\setminus D \bigcup A$.
Clearly each local step can be performed in $n^{O(t)}$ time which is polynomial for fixed $t$. The number of iterations
to reach a local optimum may be super-polynomial; however this can be made polynomial by the standard
method~\cite{AGKMMP04} of performing a  local move only if the cost $\Phi$ reduces by some $1+\frac1{poly(n)}$ factor.
Here we omit this (minor) detail and bound the local optimum under the swaps as defined above.

Let $F\sse V$ denote the local optimum solution (under $t$-swaps) and $F^*\sse V$ the global optimum. Note that
$|F|=|F^*|=k$. Define map $\eta: F^*\rightarrow F$ as $\eta(w)=\arg\min_{v\in F} d(w,v)$ for all $w\in F^*$. For any
$S\sse V$, let $\med(S) := \sum_{u\in V} q_u\cdot d(u,S)$, and $\tree(S) = d(MST(V/S))$ be the length of the minimum
spanning tree in the metric obtained by contracting $S$; so $\Phi(S) = \med(S)+\tree(S)$. For any $D\sse F$ and $A\sse
V\setminus F$ with $|D|=|A|\le t$ we refer to the swap $F-D+A$ as a ``$(D,A)$ swap''. We use the following swap
construction from~\cite{GT08} for the $k$-median problem.

\begin{theorem}[\cite{GT08}] \label{th:GT-swap}
For any $F,F^*\sse V$ with $|F|=|F^*|=k$, there are partitions $\{F_i\}_{i=1}^\ell$ of $F$ and $\{F^*_i\}_{i=1}^\ell$
of $F^*$ such that $|F_i|=|F^*_i|$ $\forall i\in[\ell]$; and there is a unique $c_i\in F_i$ (for each $i\in[\ell]$)
with $\eta(w)=c_i$ for all $w\in F^*_i$ and $\eta^{-1}(v)=\emptyset$ for all $v\in F_i\setminus \{c_i\}$.  Define set
$\swap$ of $t$-swaps with multipliers $\{\alpha(s) : s\in \swap\}$ as:
\begin{itemize}
 \item For any $i\in[\ell]$, if $|F_i|\le t$ then swap $(F_i, F^*_i)\in \swap$ with $\alpha(F_i,F^*_i)=1$.
 \item For any $i\in[\ell]$, if $|F_i|> t$ then for each $a\in F^*_i$ and $b\in F_i\setminus \{c_i\}$ swap $(b,a)\in
 \swap$ with $\alpha(b,a)=\frac1{|F_i|-1}$.
\end{itemize}
Then we have:
\begin{itemize}
 \item  $\sum_{(D,A)\in \swap} \alpha(D,A)\cdot \left( \med(F-D+A)-\med(F)\right) \le (3+2/t)\cdot \med(F^*) -
\med(F)$.
 \item For each $w\in F^*$, the extent to which $w$ is added $\sum_{(D,A)\in \swap: w\in A} \alpha(D,A)=1$.
\item For each $v\in F$, the extent to which $v$ is dropped $\sum_{(D,A)\in \swap: v\in D} \alpha(D,A)\le 1+\frac1t$.
\end{itemize}
\end{theorem}

We use the same set $\swap$ of swaps for the \kmf problem and will show the following:
\begin{equation}\label{eq:tree-swap}
\sum_{(D,A)\in \swap} \alpha(D,A)\cdot \left( \tree(F-D+A)-\tree(F)\right) \le (3+2/t)\cdot \tree(F^*) -
\tree(F)\end{equation}

Combined with the similar inequality in Theorem~\ref{th:GT-swap} (for $\med$) and using local optimality of $F$, we
would obtain the main result of this section:
\begin{theorem} The $t$-swap local search algorithm for \kmf is a $\left(3+\frac2t\right)$-approximation.\end{theorem}

It remains to prove~\eqref{eq:tree-swap}, which we do in the rest of the section. Consider a graph $H$ which is the
complete graph on vertices $V\bigcup\{r\}$ (for a new vertex $r$). If $E={V\choose 2}$ denotes the edges in the metric,
$H$ has edges $E \bigcup \{(r,v) :\, v\in V\}$; the edges $\{(r,v) :\, v\in V\}$ are called {\em root-edges} and edges
$E$ are {\em true-edges}. Let $M$ denote the {\em spanning tree} of $H$ consisting of edges $MST(V/F) \bigcup \{(r,v) :
v\in F\}$; similarly $M^*$ is the spanning tree $MST(V/F^*) \bigcup \{(r,v) : v\in F^*\}$. For ease of notation, for
any subset $S\sse V$, when it is clear from context we will use $S$ to also denote the set $\{(r,v) : v\in S\}$ of
root-edges. We start with the following exchange property (which holds more generally for any matroid), see
Equation~(42.15) in Schrijver~\cite{Schr-book}.
\begin{theorem}[\cite{Schr-book}]\label{th:mat-exch}
Given two spanning trees $T_1$ and $T_2$ in a graph $H$ and a partition $\{T_1(i)\}_{i=1}^p$ of the edges of $T_1$,
there exists a partition $\{T_2(i)\}_{i=1}^p$ of edges of $T_2$ such that $(T_2\setminus T_2(i)) \bigcup T_1(i)$ is a
spanning tree in $H$ for each $i\in[p]$. (This also implies $|T_2(i)|=|T_1(i)|$ for all $i\in[p]$).
\end{theorem}

\begin{figure}
\begin{center}
\includegraphics[scale=0.75]{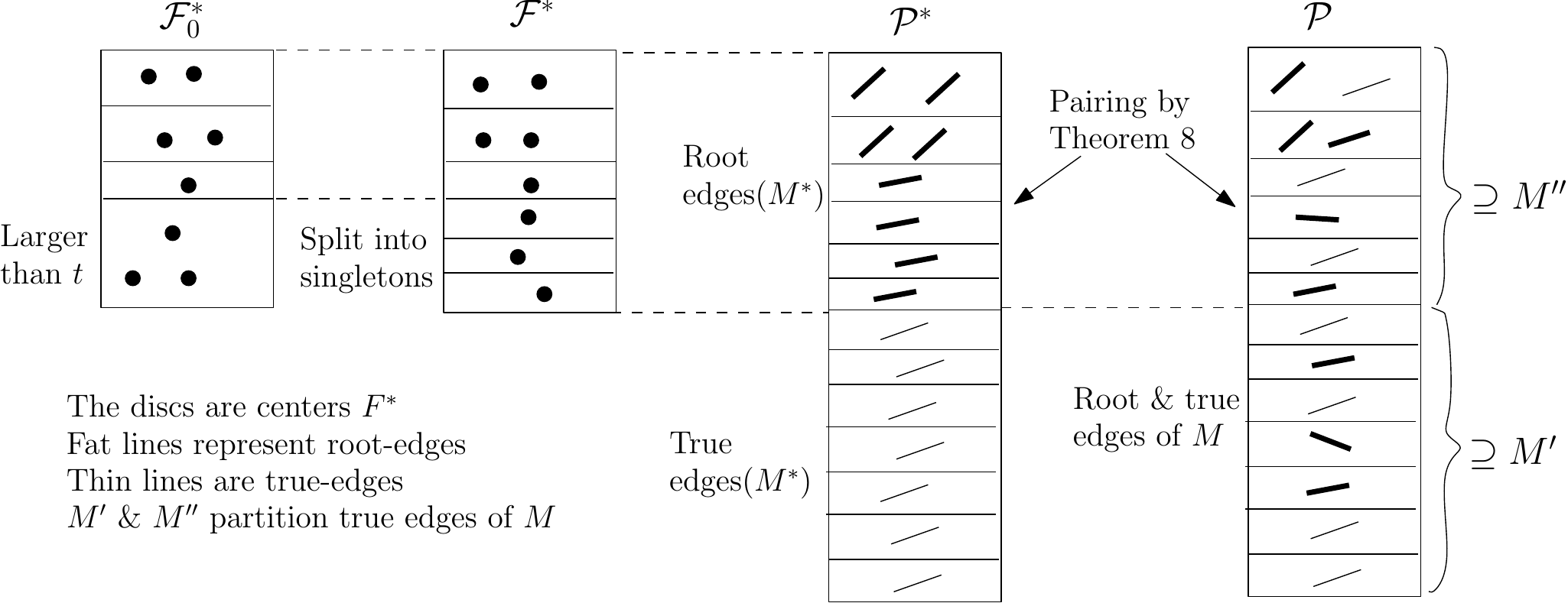}
\caption{The partitions used in local search proof (eg. has $k=8$ and $t=2$).\label{fig:partn}}
\end{center}
\end{figure}

We will apply Theorem~\ref{th:mat-exch} on trees $M^*$ and $M$. Throughout, $M^*$ and $M$ represent the corresponding
edge-sets. Recall the partition $\f^*_0 := \{F^*_i\}_{i=1}^\ell$ of $F^*$ from Theorem~\ref{th:GT-swap}; we refine
$\f^*_0$ by splitting parts of size larger than $t$ into singletons, and let $\f^*$ denote the resulting partition (see
Figure~\ref{fig:partn}). The reason behind splitting the large parts of $\{F^*_i\}_{i=1}^\ell$ is to ensure the
following property (recall swaps \swap from Theorem~\ref{th:GT-swap}).
\begin{claim}\label{cl:fstar-partn}
For each swap $(D,A)\in\swap$, $A\sse F^*$ appears as a part in $\f^*$. Moreover, for each part $A'$ in $\f^*$ there is
some swap $(D',A')\in\swap$.
\end{claim}

Consider the partition $\p^*$ of $M^*$ with parts $\f^*\bigcup \{e\}_{e\in M^*\setminus F^*}$, i.e. each true edge lies
in a singleton part  and the root edges form the partition $\f^*$ defined above. Let $\p$    
denote the partition of $M$ obtained by applying Theorem~\ref{th:mat-exch} with partition $\p^*$ of $M^*$; note also
that there is a {\em pairing} between parts of $\p$ and $\p^*$. Let $M'\sse M\cap E$ denote the true edges of $M$ that
are paired with true edges of $M^*$; and $M''=(M\cap E) \setminus M'$ are the remaining true edges of $M$ (see also
Figure~\ref{fig:partn}). We will bound the cost of $M'$ and $M''$ separately.

\begin{claim}\label{cl:kmf-m'}
$\sum_{e\in M'} d_e \le \sum_{h\in E\cap M^*} d_h$.
\end{claim}
\begin{proof} Fix any $e\in M'$. By the definition of $M'$ it follows that there is an $h\in E\cap M^*$
such that part $\{h\}$ in $\p^*$ is paired with part $\{e\}$ in $\p$. In particular,  $M-e+h$ is a spanning tree in
$H$. Note that the root edges in $M-e+h$ are exactly $F$, and so $M-e+h$ is a spanning tree in the original metric
graph $(V,E)$ when we contract vertices $F$. Since $M=MST(V/F)$ is the minimum such tree, we have $d(M)-d_e+d_h\ge
d(M)$ or $d_e\le d_h$. Summing over all $e\in M'$ and observing that each edge $h\in E\cap M^*$ can be paired with at
most one $e\in M'$, we obtain the claim.\end{proof}

Consider the connected components (in fact a forest) induced by true-edges of $M$: for each $f\in F$ let $C_f\sse V$
denote the vertices connected to $f$. Note that $\{C_f:f\in F\}$ partitions $V$.

Now consider the forest induced by true edges of $M^*$ (i.e. $E\cap M^*$) and direct each edge towards an $F^*$-vertex
(note that each tree in this forest contains exactly one $F^*$-vertex). Observe that each vertex $v\in V\setminus F^*$
has exactly one out-edge $\sigma_v$, and $F^*$-vertices have none.

For each $f\in F$, define $T_f := \{\sigma_v : v\in C_f\}$ the set of out-edges from $C_f$.
\begin{claim}\label{cl:kmf-m*} $\sum_{f\in F} d(T_f) = d(E\cap M^*)$.
\end{claim}
\begin{proof} It is clear that $\{T_f\}_{f\in F}$ partitions $E\cap M^*$.\end{proof}

\medskip
We are now ready to bound the increase in the $\tree$ cost under swaps \swap. By Claim~\ref{cl:fstar-partn} it follows
that for each swap $(D,A)\in \swap$, $A$ is a part in $\f^*$ (and so in $\p^*$); define $E_A$ as the true-edges of $M$
(possibly empty) that are paired with the part $A$ of $\p^*$.

\begin{claim}\label{cl:kmf-Ea} $\{E_A : (D,A)\in \swap\}$ is a partition of $M''$.\end{claim}
\begin{proof}
Consider the partition $\p$ of $M$ given by Theorem~\ref{th:mat-exch} applied to $\p^*$. By definition, $M'\sse E\cap
M$ are the true edges of $M$ paired (by $\p$ and $\p^*$) with true edges of $M^*$; and $M''=(E\cap M)\setminus M'$ are
paired with parts from $\f^*$ (i.e. root edges of $M^*$). For each part $\pi\in \f^*$ (and also $\p^*$) let $E(\pi)\sse
M''$ denote the $M''$-edges paired with $\pi$. It follows that $\{E(\pi) : \pi\in \f^*\}$ partitions $M''$. Using the
second fact in Claim~\ref{cl:fstar-partn} and the definition $E_A$s, we have $\{E_A : (D,A)\in \swap\} = \{E(\pi) :
\pi\in \f^*\}$, a partition of $M''$.
\end{proof}

We prove the following key lemma.

\begin{lemma}\label{lem:kmf-main}
For each swap $(D,A)\in \swap$, $\tree(F-D+A)-\tree(F)\le 2\cdot \sum_{f\in D} d(T_f) - d(E_A)$.
\end{lemma}
\begin{proof}
By Claim~\ref{cl:fstar-partn}, $A\sse F^*$ is a part in $\p^*$. Recall that $E_A$ denotes the true-edges of $M$ paired
with $A$; let $F_A$ denote the root-edges of $M$ paired with $A$. Then using Theorem~\ref{th:mat-exch} it follows that
$(M\setminus E_A \setminus F_A)\bigcup A$ is a spanning tree in $H$. Hence $S_A := (E\cap M) \setminus E_A$ is a forest
with each component containing some vertex from $F\cup A$; for any $f\in F\cup A$ let $C'_f$ denote vertices in the
component containing $f$.
In other words, $S_A$ connects connects each vertex to some vertex of $F\cup A$.


Consider the edge set $S'_A := S_A \bigcup_{f\in D} T_f$. We will add some edges $N$ so that $S'_A\bigcup N$ connects
each $D$-vertex to some vertex of $F-D+A$. Since $S_A$ already connects all vertices to $F\cup A$, it would follow that
$S'_A\bigcup N$ connects all vertices to $F-D+A$, i.e.
$$\tree(F+A-D)\le d(S'_A) + d(N) \le \tree(F) - d(E_A) +\sum_{f\in D} d(T_f)+ d(N).$$
To prove the lemma it now suffices to construct a set $N$ with $d(N)\le \sum_{f\in D} d(T_f)$, such that $S'_A\bigcup
N$ connects each $D$-vertex to $F-D+A$. Below, for any $V'\sse V$ we use $\delta(V')$ to denote the edges of $S'_A$
between $V'$ and $V\setminus V'$.

\paragraph{Constructing $N$} Consider any minimal $U\sse D$ such that $\delta\left(\bigcup_{f\in U} C'_f\right)=\emptyset$; recall that $C'$s are
the connected components of $S_A\sse S'_A$. By minimality of $U$, it follows that $\bigcup_{f\in U} C'_f$ is connected
in $S'_A$. We now prove two simple claims:


\begin{claim}\label{cl:tree-loc2} For any $f^*\in F^*\setminus A$ we have $\eta(f^*)\not\in D$.
\end{claim}
\begin{proof} By construction of the swaps \swap in Theorem~\ref{th:GT-swap}.
\end{proof}

\begin{claim}\label{cl:tree-loc3}
There exists $f^*\in F^* \bigcap \left(\bigcup_{f\in U} C'_f\right)$ and $f'\in U$ such that $\bigcup_{f\in U} T_f$
contains a path between $f'$ and $f^*$.
\end{claim}
\begin{proof} Let any $f'\in U$. Consider the directed path $P$ from $f'$ obtained by following {\em out-edges $\sigma$} until the {\em first
occurrence} of  a vertex $v\in F^*$ or $v\in V\setminus \left(\bigcup_{f\in U} C'_f\right)$. Since $F^*$-vertices are
the only ones with no out-edge $\sigma$, and $\{\sigma_w : w\in V\}=E\cap M^*$ is acyclic, there must exist such a
vertex $v\in F^* \bigcup \left( V\setminus \left(\bigcup_{f\in U} C'_f\right)\right)$. Observe that $C'_f\sse C_f$ for
all $f\in D\supseteq U$; recall that $C$s (resp. $C'$s) are the connected components in $M$ (resp. $S_A\sse M$). So
$P\sse \{\sigma_w : w\in \bigcup_{f\in U} C'_f\}\sse  \{\sigma_w : w\in \bigcup_{f\in U} C_f\} = \bigcup_{f\in U} T_f$.
Suppose that vertex $v\not\in F^*$, then $v\in V\setminus \left(\bigcup_{f\in U} C'_f\right)$ which implies
$\delta\left(\bigcup_{f\in U} C'_f\right) \ne \emptyset$ since path $P\sse S'_A$ leaves $\bigcup_{f\in U} C'_f$. So we
have $v\in F^* \bigcap \left(\bigcup_{f\in U} C'_f\right)$ and $P\sse \bigcup_{f\in U} T_f$ is a path from $f'$ to $v$.
\end{proof}

Consider $f^*$ and $f'$ as given Claim~\ref{cl:tree-loc3}. If $f^*\in A$ then the component $\bigcup_{f\in U} C'_f$ of
$S'_A$ is already connected to $F-D+A$. Otherwise by Claim~\ref{cl:tree-loc2} we have $\eta(f^*)\not\in D$; in this
case we add edge $\left(f^*,\eta(f^*)\right)$ to $N$ which connects component $\bigcup_{f\in U} C'_f$ to $\eta(f^*)\in
F-D\sse F-D+A$. Now using Claim~\ref{cl:tree-loc3}, $d\left(f^*,\eta(f^*)\right)\le d(f^*,f')\le \sum_{f\in U}
d(T_f)$.\footnote{This is the only place in the proof where we use uniformity in the metrics for $k$-median and MST.}
In either case, $U$ is connected to $F-D+A$ in $S'_A\bigcup N$, and cost of $N$ increases by at most $\sum_{f\in U}
d(T_f)$.

We apply the above argument to {\em every minimal} $U\sse D$ with $\delta\left(\bigcup_{f\in U} C'_f\right)=\emptyset$.
The increase in cost of $N$ due to each such $U$ is at most $\sum_{f\in U} d(T_f)$. Since such minimal sets $U$s are
disjoint, we have $d(N)\le \sum_{f\in D} d(T_f)$. Clearly $S'_A\bigcup N$ connects each $D$-vertex to $F-D+A$.
\end{proof}
\medskip

Using this lemma for each $(D,A)\in\swap$ weighted by $\alpha(D,A)$ (from Theorem~\ref{th:GT-swap}) and adding,
\begin{eqnarray}
&& \sum_{(D,A)\in \swap} \alpha(D,A)\cdot \left[ \tree(F-D+A)-\tree(F) \right] \notag \\
&\le &2\cdot \sum_{(D,A)\in \swap} \alpha(D,A)\cdot \sum_{f\in D} d(T_f) - \sum_{(D,A)\in \swap} \alpha(D,A)\cdot d(E_A) \label{eq:kmf-ls1}\\
& = & 2\, \sum_{f\in F} \left( \sum_{(D,A)\in\swap : f\in D} \alpha(D,A) \right) \cdot d(T_f) - \sum_{e\in M''} \left(
\sum_{(D,A)\in\swap : e\in E_A} \alpha(D,A) \right)\cdot d_e \label{eq:kmf-ls2}\\
& \le & 2\left(1+\frac1t\right) \, \sum_{f\in F} d(T_f) - \sum_{e\in M''} d_e \label{eq:kmf-ls3}\\
&=& 2\left(1+\frac1t\right)\cdot d(E\cap M^*) - d(M'') \label{eq:kmf-ls4}
\end{eqnarray}
Above~\eqref{eq:kmf-ls1} is by Lemma~\ref{lem:kmf-main},~\eqref{eq:kmf-ls2} is by interchanging summations using the
fact that $E_A\sse M''$ (for all $(D,A)\in\swap$) from Claim~\ref{cl:kmf-Ea}. The first term in~\eqref{eq:kmf-ls3} uses
the property in Theorem~\ref{th:GT-swap} that each $f\in F$ is dropped (i.e. $f\in D$) to extent at most $1+\frac1t$;
the second term uses the property in Theorem~\ref{th:GT-swap} that each $f^*\in F^*$ is added to extent one in \swap
and Claim~\ref{cl:kmf-Ea}. Finally~\eqref{eq:kmf-ls4} is by Claim~\ref{cl:kmf-m*}.

Adding the inequality $0\le d(E\cap M^*) - d(M')$ from Claim~\ref{cl:kmf-m'} yields:
\begin{equation*}\sum_{(D,A)\in \swap}
\alpha(D,A)\cdot \left[ \tree(F-D+A)-\tree(F) \right] \le \left(3+\frac2t\right)\cdot d(E\cap M^*) - d(E\cap
M),\end{equation*} since $M'$ and $M''$ partition the true edges $E\cap M$. Thus we obtain
Inequality~\eqref{eq:tree-swap}.

\section{Non-uniform $k$ median forest}\label{sec:non-unif-kmf}
In this section we study the following extension of $k$ median forest. There is a set of vertices $V$ with weights
$\{q_u\}_{u\in V}$, two metrics $d$ and $c$ defined on $V$, and a bound $k$. The goal is to find $S\sse V$ with $|S|=k$
minimizing $\sum_{u\in V} q_u\cdot d(u,S) \, + \, c\big(\,\mbox{MST}(V/S)\,\big)$. Here $\mbox{MST}(V/S)$ is a minimum
spanning tree in the graph obtained by contracting $S$ to a single vertex under metric $c$. The difference from the
\kmf problem is that the cost functions for the $k$-median and MST parts in the objective are different.

It is natural to consider the local search algorithm in this setting as well, since local search achieves good
approximations for both $k$-median and MST. However the next lemma shows that the locality gap is unbounded even if we
allow multiple swaps. The example is similar to the locality gap in~\cite{KKNSS11}.
\begin{lemma}
The locality gap of non-uniform \kmf with multi-swaps is unbounded.
\end{lemma}
\begin{proof}
Fix values $M\gg w\gg 1$. Let $V=\left\{ u_{i,j} : i\in[k],\, j\in\{1,2\}\right\}$, so $|V|=2k$. Define vertex-weights
as follows: $q(u_{k,2})=1$ and all other vertices have weight $w$. The metric $d$ for the $k$-median part is:
$$
d(x,y) = \left\{
\begin{array}{ll}
0 & \mbox{ if either }x=y \mbox{ or }\{x,y\} = \{u_{i,2},u_{i+1,1}\} \mbox{ for some }i\in [k-1]\\
1 & \mbox{ otherwise} \end{array}\right.
$$
The second metric $c$ for the MST part of the objective is:
$$
c(x,y) = \left\{
\begin{array}{ll}
0 & \mbox{ if either }x=y \mbox{ or }\{x,y\} = \{u_{i,1},u_{i,2}\} \mbox{ for some }i\in [k]\\
M & \mbox{ otherwise} \end{array}\right.
$$
Observe that for any $S\sse V$ with $|S|=k$, we have $c(MST(V/S))<M$ iff $|S\bigcap \{u_{i,1}, u_{i,2}\}|=1$ for all
$i\in[k]$. So the non-uniform \kmf objective is smaller than $M$ only if $|S\bigcap \{u_{i,1}, u_{i,2}\}|=1,\, \forall
i\in[k]$.

We claim that the optimal value is at most one. Consider the solution $S^*=\{u_{i,1}\}_{i=1}^k$. It is clear that
$c(MST(V/S^*))=0$. Moreover, $\sum_{u\in V} q(u)\cdot d(u,S^*) = 1$ with vertex $u_{k,2}$ being the only contributor.

We now claim that the solution $L=\{u_{i,2}\}_{i=1}^k$ is locally optimal under even $(k-1)$-swaps. First, observe that
$c(MST(V/L))=0$ and $\sum_{u\in V} q(u)\cdot d(u,S^*) = w$ with vertex $u_{1,1}$ being the only contributor. So $L$ has
objective value of $w$. Secondly, notice that every solution $S$ obtained by some $(k-1)$-swap of $L$ has either
MST-objective of $M$ or median-objective of $w$. Thus $L$ is a local optimum and the locality gap is $w\gg 1$.
\end{proof}

\medskip

We remark that the near-optimality proof of local search in the previous section only requires the following
consistency property between the two metrics: for any pair $e,f$ of edges $d_e\le d_f \implies c_e\le c_f$. In spite of
the large locality gap, we show that non-uniform \kmf admits a constant factor approximation algorithm via an LP
approach.

\def\spp{\mathbb{SP}}
\newcommand{\pri}{{\mathcal{P}}}
\newcommand{\lpo}{\ensuremath{\mathsf{LP_{med}}\xspace}}
\newcommand{\LP}{\ensuremath{\mathsf{LP}\xspace}}
\newcommand{\I}{\ensuremath{{\mathcal{I}}}\xspace}
\newcommand{\J}{\ensuremath{{\mathcal{M}}}\xspace}

\paragraph{The algorithm.} We make use of the following natural LP relaxation for non-uniform \kmf. The variables $y_v$ denote the probability
of locating a depot at $v$; $x_{uv}$ denotes the extent to which vertex $u$ is connected to a depot at $v$ (for the
$k$-median part); and $z_e$ denotes the extent to which edge $e$ is used in the MST part of the objective. Also
$E={V\choose 2}$ is the set of all edges in the metric. Define $H$ to be the complete graph on vertices $V\bigcup\{r\}$
(for a new vertex $r$) with edges $E \bigcup \{(r,v) :\, v\in V\}$.
\begin{alignat}{2}
  \mbox{minimize } \sum_{u \in V} q_u \cdot \sum_{v \in V} d(u,v) x_{uv} & + \sum_{e\in E} c_e\cdot z_e & &
  \tag{$\mathsf{LP}$} \\
  \mbox{subject to } \sum_{v \in V} x_{uv} &= 1 & \qquad & \forall \,
  u \in V   \label{eq:1} \\
  x_{uv} &\leq y_{v} & \qquad & \forall \, u \in
  V, v \in V \label{eq:2} \\
  \sum_{v \in V} y_{v} &\leq k  & \qquad & \label{eq:3}\\
  (y,z) &\in \spp(H)& \qquad & \label{eq:4}\\
  x_{uv}, y_{v}, z_e &\geq 0 & \qquad & \forall \, u, v \in V,\,\,\, \forall e\in E \label{eq:5}
\end{alignat}
Above $\spp(H)$ denotes the spanning tree polytope of graph $H$, which admits a linear description in terms of its edge
variables; see eg.~\cite{Schr-book}. Also $(y,z) \in \spp(H)$ corresponds to the fractional spanning tree in $H$ with
values $z_e$ on edges $e\in E$ and value $y_v$ on each edge $(r,v)$. It can be checked directly that this is a valid
relaxation of non-uniform \kmf. Moreover this LP can be solved exactly in polynomial time to obtain solution
$(x^*,y^*,z^*)$ using the Ellipsoid algorithm.

We now describe the rounding procedure. Let \I denote the instance of \kmf and $\lpo = \sum_{u \in V} \sum_{v \in V}
d(u,v) x^*_{uv}$ denote the median part of the optimal LP solution. Apply Stage I of the rounding algorithm
in~\cite{KKNSS11} to modify variables $x^*$ to $\overline{x}$ (here $y^*$ and $z^*$ remain unchanged), with the
following properties:
\begin{OneLiners}
 \item Set $R\sse V$ of representatives with weights $w_u$ for each $u\in R$, which defines a new instance \J of non-uniform \kmf (the weights of
 vertices in $V\setminus R$ are zero).
 \item Any solution to the new instance \J with objective $C$ is a solution to the original instance \I having
 objective at most $C+4\lpo$.
 \item $(\overline{x},y^*,z^*)$ is feasible for $\LP(\J)$.
 \item Disjoint collection of subsets $\{\pri(u)\sse V\}_{u\in R}$ with $\sum_{v \in \pri(u)} y_{v}\ge \frac12$ for all $u\in R$.
 \item Collection of pseudoroots $\{(a_i,b_i)\in {R\choose 2}\}_{i=1}^t$ with each representative in at
most one pseudoroot.
 \item Map $\sigma:R\rightarrow R$ where $\sigma(u)$ lies in a pseudoroot for each $u\in R$.
 \item Each $u\in R$ is connected (under $\overline{x}$) only to $\pri(u)\cup \{\sigma(u)\}$.
 \item $\sum_{u\in R} \,\,w_u \cdot \left[ \sum_{v\in \pri(u)} d_{u,v}\cdot \overline{x}_{u,v} \, + \, d_{u,\sigma(u)}\cdot
 \left(1- \sum_{v\in \pri(u)} \overline{x}_{u,v} \right) \right] \le 4\cdot \lpo$.
\end{OneLiners}

\medskip
Now apply the LP reformulation from Stage II in~\cite{KKNSS11} to eliminate $x$-variables in \LP, using the above
structure of $(\overline{x},y,z)$, and obtain: \begin{small}\begin{alignat}{2}
  \mbox{minimize } \sum_{u\in R} \,\,w_u \cdot \bigg[  \sum_{v\in \pri(u)} d_{u,v}\cdot y_v \,  + \, d_{u,\sigma(u)}\cdot
& \left(1- \sum_{v\in \pri(u)} y_v \right) \bigg]  + & \sum_{e\in E} c_e\cdot z_e  &  \tag{$\mathsf{LP_{new}}$} \\
  \mbox{subject to } \sum_{v \in \pri(u)} y_{v} &\le 1 & \qquad & \forall \,   u \in R   \label{nlp:1} \\
  \sum_{v \in \pri(a_i)} y_{v} + \sum_{v \in \pri(b_i)} y_{v} &\geq 1 & \qquad & \forall \, \mbox{pseudoroots} (a_i,b_i) \label{nlp:2} \\
  \sum_{v \in V} y_{v} &\leq k  & \qquad & \label{nlp:3}\\
  (y,z) &\in \spp(H)& \qquad & \label{nlp:4}\\
  y_{v}, z_e &\geq 0 & \qquad & \forall \, v \in V,\,\,\, \forall e\in E \label{nlp:5}
\end{alignat}\end{small}

Based on the above properties, it follows that $(y^*,z^*)$ is a feasible solution to $\mathsf{LP_{new}}$ with objective
at most $4\cdot \lpo + c\cdot z^*$, i.e. at most four times the optimal value of $\LP(\I)$. The advantage of the new LP
is:
\begin{lemma}
Any basic feasible solution to $\mathsf{LP_{new}}$ is integral.
\end{lemma}
\begin{proof}
Let $(y,z)$ denote any basic feasible solution. The constraints from  \eqref{nlp:1}-\eqref{nlp:3} define a laminar
family on just $y$ variables. By a standard uncrossing argument, we can choose a maximum linearly independent set of
tight rank constraints in~\eqref{nlp:4} to be a chain on $y,z$ variables. Thus a maximum linearly independent set of
tight constraints in $\mathsf{LP_{new}}$ can be described as the intersection of two laminar families-- this is always
a totally unimodular matrix, and hence $(y,z)$ must be integral.
\end{proof}

$\mathsf{LP_{new}}$ can be solved exactly in polynomial time to obtain an extreme point solution using the Ellipsoid
algorithm and the approach in Jain~\cite{J01}; by the above lemma this solution is integral. Finally using Lemma 3.3
in~\cite{KKNSS11}, any integral solution to $\mathsf{LP_{new}}$ of value $L$ is also a valid solution to the \kmf
instance \J of value at most $3\cdot L$. Altogether we obtain an integral solution $S^*$ to \J of value at most 12
times the optimum of $\LP(\I)$. Combined with the relation between instances \I and \J, we have $S^*$ is a valid
solution to \I of objective at most 16 times the optimum of \I, thereby proving Theorem~\ref{th:gen-kmed-forest}.

\bibliography{lvrp}

\begin{thebibliography}{10}

\bibitem{AG87}
K.~Altinkemer and B.~Gavish.
\newblock Heuristics for unequal weight delivery problems with a fixed error
  guarantee.
\newblock {\em Operations Research Letters}, 6:149--158, 1987.

\bibitem{AG90}
K.~Altinkemer and B.~Gavish.
\newblock Heuristics for delivery problems with constant error guarantees.
\newblock {\em Transportation Research}, 24:294--297, 1990.

\bibitem{AGKMMP04}
Vijay Arya, Naveen Garg, Rohit Khandekar, Adam Meyerson, Kamesh Munagala, and
  Vinayaka Pandit.
\newblock Local search heuristics for k-median and facility location problems.
\newblock {\em SIAM J. Comput.}, 33(3):544--562, 2004.

\bibitem{BWW87}
A.~Balakrishnan, J.E. Ward, and R.T. Wong.
\newblock Integrated facility location and vehicle routing models: Recent work
  and future prospects.
\newblock {\em American Journal of Mathematical and Management Sciences},
  7:35--61, 1987.

\bibitem{BJS95}
O.~Berman, P.~Jaillet, and D.~Simchi-Levi.
\newblock Location-routing problems with uncertainty.
\newblock In Z.~Drezner, editor, {\em Facility Location: A Survey of
  Applications and Methods}, pages 427--452. Springer, 1995.

\bibitem{CGTS99}
Moses Charikar, Sudipto Guha, \'{E}va Tardos, and David~B. Shmoys.
\newblock A constant-factor approximation algorithm for the k-median problem.
\newblock {\em J. Comput. Syst. Sci.}, 65(1):129--149, 2002.

\bibitem{GT08}
Anupam Gupta and Kanat Tangwongsan.
\newblock Simpler analyses of local search algorithms for facility location.
\newblock {\em CoRR}, abs/0809.2554, 2008.

\bibitem{HK85}
M.~Haimovich and A.~H. G.~Rinnooy Kan.
\newblock Bounds and heuristics for capacitated routing problems.
\newblock {\em Mathematics of Operations Research}, 10(4):527--542, 1985.

\bibitem{HKM10}
T.~Harks, F.G. K{\"o}nig, and J.~Matuschke.
\newblock Approximation algorithms for capacitated location routing.
\newblock {\em T.U. Berlin Preprint 010-2010}.

\bibitem{J01}
Kamal Jain.
\newblock A factor 2 approximation algorithm for the generalized steiner
  network problem.
\newblock {\em Combinatorica}, 21(1):39--60, 2001.

\bibitem{JV01}
Kamal Jain and Vijay~V. Vazirani.
\newblock Approximation algorithms for metric facility location and k-median
  problems using the primal-dual schema and lagrangian relaxation.
\newblock {\em J. ACM}, 48(2):274--296, 2001.

\bibitem{KKNSS11}
R.~Krishnaswamy, A.~Kumar, V.~Nagarajan, Y.~Sabharwal, and B.~Saha.
\newblock The matroid median problem.
\newblock In {\em SODA}, 2011.

\bibitem{L88}
G.~Laporte.
\newblock Location-routing problems.
\newblock In B.L. Golden and A.A. Assad, editors, {\em Vehicle Routing: Methods
  and Studies}, pages 163--198. North-Holland, 1988.

\bibitem{L89}
G.~Laporte.
\newblock A survey of algorithms for location-routing problems.
\newblock {\em Investigaci{\'o}n Operativa}, 1:93--123, 1989.

\bibitem{LS90}
C.~Li and D.~Simchi-Levi.
\newblock Worst-case analysis of heuristics for multidepot capacitated vehicle
  routing problems.
\newblock {\em ORSA Journal on Computing}, 2(1):64--73, 1990.

\bibitem{MJS98}
H.~Min, V.~Jayaraman, and R.~Srivastava.
\newblock Combined location-routing problems: A synthesis and future research
  directions.
\newblock {\em European Journal of Operational Research}, 108:1--15, 1998.

\bibitem{NS07}
G.~Nagy and S.~Salhi.
\newblock Location-routing: Issues, models and methods.
\newblock {\em European Journal of Operational Research}, 177(2):649--672,
  2007.

\bibitem{Schr-book}
A.~Schrijver.
\newblock {\em {C}ombinatorial {O}ptimization}.
\newblock Springer, 2003.

\bibitem{TV02}
Paolo Toth and Daniele Vigo, editors.
\newblock {\em The vehicle routing problem}.
\newblock SIAM Monographs on Discrete Mathematics and Applications,
  Philadelphia, PA, USA, 2002.

\end{thebibliography}
\bibliographystyle{plain}

\appendix

\section{Example comparing $k$-median, $k$-tree and \kmf}\label{app:example}
We give an example which shows that near-optimal solutions to the $k$-median, $k$-tree and \kmf problems can be very
far from each other. This implies that an approximation algorithm for \kmf must simultaneously take into account both
the median and tree parts of its objective. (For eg. we cannot merely solve $k$-median and $k$-tree separately and take
the better of those solutions.)

The underlying metric consists of six vertices $\{u_0,u_1,u_2\}\bigcup \{v_0,v_1,v_2\}$. Let $\ell$ be a parameter that
will be set to be arbitrarily large. The distance between any $u_i$ and $v_j$ (for all $i,j\in\{0,1,2\}$) is infinite;
$d(u_0,u_1)=d(u_0,u_2)=\ell^3$, $d(u_1,u_2)=\ell^2$; and $d(u_0,u_1)=d(u_0,u_2)=\ell^4$, $d(u_1,u_2)=\ell$. The weights
of vertices are $q(u_1)=q(u_2)=q(v_1)=q(v_2)=\ell^4$ and $q(u_0)=q(v_0)=1$. The bound $k=4$ and parameter $\rho=\ell^2$
for the \kmf problem. Let $S_{med}$, $S_{tree}$ and $S_{kmf}$ denote solutions that are $o(\ell)$-approximately optimal
for the $k$-median, $k$-tree and \kmf objectives respectively. We claim that $S_{med}$, $S_{tree}$ and $S_{kmf}$ are
mutually disjoint.

It can be checked directly that the optimal $k$-median value is $\ell^3+\ell^4\le 2\,\ell^4$. Moreover the only
solution of value $o(\ell^5)$ is $\{u_1,u_2,v_1,v_2\}$; so $S_{med}$ consists of just this solution.

The optimal $k$-tree value is $\ell+\ell^2\le 2\,\ell^2$. For any solution $F\in S_{tree}$ (i.e. having value
$o(\ell^3)$), we must have $u_0,v_0\in F$, $|F\cap \{u_1,u_2\}|=1$ and $|F\cap \{v_1,v_2\}|=1$. So $S_{tree}$ consists
of these 4 solutions.

For the \kmf objective it can be seen that the optimal value is $\rho\cdot \ell^3 + \ell^4\cdot \ell=2\,\ell^5$; from
the solutions $\{u_1,u_2,v_0,v_1\}$ and $\{u_1,u_2,v_0,v_2\}$. Moreover, any other solution has value $\Omega(\ell^6)$;
so $S_{kmf}$ consists of the above two solutions. Clearly $S_{med}$, $S_{tree}$ and $S_{kmf}$ are disjoint.

\end{document}